\documentclass[a4paper,USenglish]{lipics-v2016}

\usepackage{microtype}

\usepackage{doc}
\usepackage{url}
\usepackage{enumerate}
\usepackage{amsmath,amssymb,amsthm,mathrsfs,amsfonts,dsfont}
\usepackage{graphicx}
\usepackage{verbatim}
\usepackage{alltt}
\usepackage{pdfpages}
\usepackage{subcaption}
\usepackage{algorithm}
\usepackage{caption}
\usepackage{algorithmic}
\usepackage{multirow}
\usepackage[titletoc,toc,page]{appendix}
\usepackage[colorinlistoftodos,prependcaption]{todonotes}

\newtheorem{prop}[theorem]{Proposition}

\newcommand{\lin}{linearizability}
\newcommand{\linb}{linearizable}

\newcommand{\Lin}{Linearizability}
\newcommand{\linz}{linearization}
\newcommand{\cmplt}{\textbf{cmplt}}
\newcommand{\myinv}{\textbf{inv}}
\newcommand{\myresp}{\textbf{resp}}
\newcommand{\mym}{\textbf{m}}
\newcommand{\Linz}{Linearization}
\newcommand{\his}{\text{H}}
\newcommand{\shis}{\text{S}}
\newcommand{\evts}{\textit{E}}
\newcommand{\myvert}{\vert}
\newcommand{\mybigvert}{ \ \vert \ }
\newcommand{\getreg}{\textbf{get-regiters}}
\newcommand{\objlevel}{\textbf{obj-level}}

\newcommand\blankfootnote[1]{%
	\let\thefootnote\relax\footnotetext{#1}%
	\let\thefootnote\svthefootnote%
}

\title{A Constructive Proof on the Compositionality of Linearizability}
\titlerunning{A Constructive Proof On the Compositionality of Linearizability}

\author[1]{Haoxiang Lin}
\affil[1]{
  Microsoft Research,
  Beijing, China\\
  \texttt{haoxlin@microsoft.com}
}
\authorrunning{Haoxiang Lin}

\Copyright{Haoxiang Lin}

\begin{document}

\maketitle

\begin{abstract}\label{sec:abs}	
	{\Lin} is the strongest correctness property for both shared memory and message passing systems.
	One of its useful features is the compositionality:
	a history (execution) is {\linb} if and only if each object (component) subhistory is {\linb}.
	In this paper, we propose a new hierarchical system model to address challenges in modular development of cloud systems.
	Object are defined by induction from the most fundamental atomic Boolean registers,
	and histories are represented as countable well-ordered structures of events
	to deal with both finite and infinite executions.
	Then, we present a new constructive proof on the
	compositionality theorem of {\lin} inspired by Multiway Merge.
	This proof deduces a theoretically efficient algorithm which generates {\linz}
	in $\mathcal{O}(\text{N}\times\log{}\text{P})$ running time with $\mathcal{O}(\text{N})$ space,
	where P and N are process/event numbers respectively.
\end{abstract}

\newpage

\section{Introduction}\label{sec:intro}

{\Lin}~\cite{Herlihy90} is the strongest correctness property for both shared memory and message passing systems.
Informally, a piece of process execution (e.g. a method call)
is {\linb} if it appears to take effect \textit{instantaneously} at some moment during the lifetime.
This implies that a concurrent {\linb} execution by multiple processes
should produce exactly the same result by some single-process sequential execution.
Thus, {\lin} is \textit{compositional}: a history (execution) is {\linb} if and only if each object (component) subhistory is {\linb}.
To prove this property, previous work~\cite{Herlihy86, Herlihy87, Herlihy90, Herlihy08} claims that there exists a global partial ordering
\footnote
{
	When referring to an ordering, we assume the strictness in that the ordering is irreflexive.
}
among {\linz}s of object subhistories
\footnote
{
	Proofs in~\cite{Herlihy86, Herlihy87, Herlihy08} are not sound.
	In~\cite{Herlihy86, Herlihy87}, the \textit{inter-object} partial ordering is omitted by mistake.
	In~\cite{Herlihy08}, a case is left out that there could be no maximal method call among the last method calls from object subhistories.
}
, and then extends it to the final {\linz}(s) by Order-Extension Principle~\cite{Szpilrajn30}.
Compositionality is very important because the correctness of a concurrent system can rely on its components instead of a centralized scheduler or additional component constraints~\cite{Herlihy90, Herlihy08}.

Recently, new challenges are emerging in modular development of cloud systems.
People usually design, implement and test a complex system from the bottom up with existing components layer by layer.
However, previous system model and compositionality result on {\lin} assume flat components and are not able to handle such \textit{hierarchy} well, leaving a gap between theory and practice.
In addition, previous studies on {\lin} normally investigate finite executions so that infinite executions
\footnote
{
	Paper~\cite{DBLP:conf/netys/GuerraouiR14} proves that {\lin} may not be always a safety property on infinite histories.
	We assume objects with finite nondeterminism to resolve the problem.
}
of long running cloud services are ruled out unfortunately.
The last challenge lies in system verification to provide a concrete {\linz} from component executions.
A typical algorithm traverses component {\linz}s to build the global partial ordering (e.g. represented by a directed acyclic graph), and then uses topological sorting~\cite{Cormen} for a final {\linz}.
Because a cloud system has lots of processes/threads and their executions take long enough,
such algorithm is considerably time consuming.

To address the above challenges,
we begin by proposing a universal model for hierarchical systems in section~\ref{sec:foundation},
which demonstrates theoretical benefits on concurrent properties not limited to {\lin}.
An object is defined by induction from the most fundamental \textit{atomic Boolean registers} with only six atomic primitive operations.
An object operation is defined similarly by induction from subobjects' operations following the syntax of a simple imperative programming language.
Histories are then defined as \textit{countable well-ordered structures} of events
to deal with both finite and infinite executions, and reveal precisely the internal event ordering.
Such history definition over a simple sequence by previous work has advantages of simplicity and extensibility 
in later discussion on history properties and operations.
Section~\ref{sec:def_lin} defines {\lin} using our new system model,
and we make a remedy to the original definition by excluding an intricate case.
In section~\ref{sec:comp},
we present a constructive proof on the compositionality theorem with our new methodology.
The proof deduces a Multiway Merge alike algorithm to generate a final {\linz} without the need to build the global partial ordering.
Section~\ref{sec:discussion} concludes and discusses future work.

\section{System Model}\label{sec:foundation}

\subsection{Overview}\label{sec:model_overview}

In this subsection, we briefly state our system model and recall notions that we use in the paper.
Formal definitions follow later.

Registers are historical names for shared memory locations.
For simplicity, we use only atomic Boolean registers~\cite{Lamport86-1, Lamport86-2}
with six atomic primitive operations for two reasons.
First, registers of other primitive data types such as 32-bit signed integer can be represented by Boolean registers in theoretical unary notation.
We implement any $M$-valued register as an array of $M$ Boolean registers~\cite{Herlihy08},
although it is astonishingly inefficient.
Second, for Boolean-size (1 byte) memories, nowadays computer architectures already offer atomic read/write semantics as well as other atomic instructions for Compare-And-Swap, Or etc.

Processes are single-threaded and exchange information with each other through shared objects (Definition \ref{def_object}) in parallel.
From a pedantic viewpoint, processes can be thought as instances of program graphs~\cite{Pnueli92}.
Each object has some non-overlapped and persistent registers for keeping values,
an associated type defining the value domain,
and some primitive operations as the only interfaces for object creation and manipulation.
A concurrent system is thus viewed as a \textit{finite} set of processes and objects.

Operations can be mapped to actions~\cite{Kropf99}.
Their executing instances are denoted as method calls.
We assume that object operation set is also finite,
and require processes carry out one method call by another until the process itself or the system halts.
This implies that processes issue a new method call only after all previous ones complete.
The set of method calls issued by a process is assumed countable to model infinite executions.
From the above description, we can separately number processes, objects, operations, and method calls with integers.

During the lifetime of a method call, we assume that unique events are generated.
E.g. there is an \textit{invocation} event at the very beginning when the method call starts,
and a \textit{response} event right at the moment it finishes;
while in the middle some others such as a \textit{print} event are emitted.
An event could be modeled as a triple of program location, action, and evaluation of the belonging object.
Only invocation and response events are considered
because they are enough to completely depict the behavior of method calls.

Formally, we characterize invocation and response events by five key factors:
in which process it is generated, which method call it belongs to,
on which object the method call manipulates, operation id, and the associated payloads.
An invocation event is then formalized as $inv \langle i, j, x, y, args^\ast \rangle $
, while a response event as $resp \langle i, j, x, y, term(res^\ast) \rangle $.
$i, j, x, y$ are all natural numbers.
If some event factors (e.g. payloads) are inessential to our discussion,
``\textunderscore'' will be used instead.
Explicit event types are for convenience, and they could be encoded in the payloads in fact.
Such definition expresses that the event belongs to the $j$-th method call
issued by process $p_i$ on the operation $op_y$ of object $o_x$.
$args^\ast$ and $res^\ast$ stand for arguments and results respectively
while $term$ is termination condition.
An invocation event $inv \langle i, j, x, y, args^\ast \rangle$ and response event $resp \langle r, s, u, v, term(res^\ast) \rangle$
are \textit{matching} if $(i = r) \land (j = s) \land (x = u) \land (y = v)$.

\subsection{Language Syntax}\label{sec:syntax}
Suppose we implement our system in a simple imperative programming language.
An operation is just a program statement or normally a finite sequence of statements.
Here is the language \textit{abstract syntax}:
		\begin{align*}
		b \in \textbf{BExp} \ &::= \ \textbf{true} \ | \ \textbf{false} \ | \ \neg b \ | \ b_1 \lor b_2 \ | \ b_1 \land b_2\\
		ST \in \textbf{Stmt} \ &::= \ [o.op] \ | \ ST_1;ST_2 \ | \ \textbf{if} \ [b] \ \textbf{then} \ ST_1 \ \textbf{else} \ ST_2  \ | \ \textbf{while} \ [b] \ \textbf{do} \ ST_1
		\end{align*}
$b$ is a Boolean expression, and $[o.op]$ represents the function call of operation $op$ on object $o$.
Arithmetic expressions are not listed separately since they correspond to object operations.
The abstract syntax serves for later inductive definition of object/operation (\ref{sec:object}).

\subsection{Object}\label{sec:object}
Objects are defined by induction from atomic Boolean registers.
\begin{definition}\label{def_object}
	Say that $o$ is an object if it is
	\begin{enumerate}[(1)]
		\item
		a \textit{register} object: $\langle r, \{
		\text{\sc{Read}}, \text{\sc{Write}}, \text{\sc{CompAndSwap}}
		, \text{\sc{And}}, \text{\sc{Or}}, \text{\sc{Not}}
		\} \rangle$.
		\begin{enumerate}[(a)]
			\item
			$r$ is an atomic Boolean register.
			\item
			$\text{\sc{Read}}/\cdots/\text{\sc{Not}}$
			are the only six atomic and primitive operations
			\footnote
			{
				Other operands of $\text{\sc{CompAndSwap}}, \text{\sc{And}}, \text{and } \text{\sc{Or}}$ are passed as arguments.
			}
			.
		\end{enumerate}
		\item
		a \textit{composite} object: $\langle \{o_{i_1}, o_{i_2}, \cdots, o_{i_n}\}, \{op_{j_1}, op_{j_2}, \cdots, op_{j_k}\}\rangle$.
		\begin{enumerate}[(a)]
			\item
			$\{o_{i_1}, o_{i_2}, \cdots, o_{i_n}\}$ is a finite set of objects which are called $o$'s subobjects.
			This set varies for each $o$.
			\item
			$\{op_{j_1}, op_{j_2}, \cdots, op_{j_k}\}$ is a finite set of operations,
			and each is a finite sequence of one or more program statements
			(\ref{sec:syntax}) on $o_{i \in [1, n]}$ following the syntax of previously defined language.
		\end{enumerate}
	\end{enumerate}	
\end{definition}

We define a function {\objlevel} on objects by recursion to return their levels:
\begin{align*}
	{\objlevel}(o) = 
	\begin{cases}
		0, & \mbox{if } o \text{ is a register object} \\
		\displaystyle{\max\{{\objlevel}(o_1), \cdots, {\objlevel}(o_n)\} + 1}, & \mbox{otherwise}
	\end{cases}
\end{align*}

Although there may be many objects in a system, we are more interested in those \textit{shared} ones,
which are not subobjects of others and should be accessed by at least two processes in theory.
This excludes stack and thread-local objects since their changes are not externally visible.

As mentioned earlier, a concurrent system is a finite set of processes and shared objects.
When we refer to a system state,
it is not the mathematical state defined in a formal state machine~\cite{Baier08}.
Instead, in this paper we denote it as the instantaneous snapshot~\cite{Afek:1993:ASS:153724.153741} of those shared objects' encapsulated registers,
returned by the following recursive function:
\begin{align*}
{\getreg}(o) = 
	\begin{cases}
		\{r\}, & \mbox{if } o \text{ is a register object} \\
		\displaystyle{\bigcup_{i=1}^{n}} \, {\getreg}(o_i), & \mbox{otherwise}
	\end{cases}
\end{align*}
Shared objects are additionally non-overlapped
so as to strictly limit manipulation only on predefined interfaces.
That is, for every two different shared objects $o, o^\prime$, ${\getreg}(o) \bigcap {\getreg}(o^\prime) = \emptyset$.

\subsection{History}\label{sec:history}

Events in an execution have inherent causal relationship with each other.
E.g. An invocation event and its matching response event are intra-process cause and effect.
If process $p_i$ sends a message to $p_j$, the response event of \textsc{Send} method call is the inter-process cause of the invocation event of \textsc{Receive}.

We have three assumptions on events and their causal relationship:
\begin{enumerate}[(1)]
	\item
	Each event is generated by only one process.
	\item
	The causal relationship is a partial ordering,
	which means if event $e_1$ is the cause of $e_2$ and $e_2$ is the cause of $e_3$, then $e_1$ is the cause of $e_3$.
	\item
	Events from the same process are countable, well-ordered under the causal relationship and thus isomorphic to a subset of natural numbers.
\end{enumerate}

Due to the finiteness of processes, this causal relationship is in fact a partial well ordering because of its well-foundedness~\cite{Enderton77} that every nonempty event subset contains a causal-minimal element.
To \textit{compare} each two events and find a \textit{least} one from any nonempty event subset,
we extend it to a well ordering
by the following Lemma \ref{par_to_well}
\footnote
{
	For extending arbitrary partial well orderings, please refer to~\cite{Lin-well2015}.
}
.
In practice, if the system has a global high-precision clock,
the event generation time could be leveraged.

\begin{lemma}\label{par_to_well}
	Assume $\langle \evts, \prec \rangle$ is a countable partially-ordered structure that:
	\begin{enumerate}[(1)]
		\item
		There exists a finite partition $\Pi = \{\evts_1, \evts_2, \cdots, \evts_m\}$ of $\evts$.
		\item
		For each $\evts_i \in \Pi$, $\prec \bigcap \ (\evts_i \times \evts_i)$ is a well ordering. 
	\end{enumerate}
	Then $\prec$ can be extended to a well ordering.
\end{lemma}
\begin{proof}
	By Order-Extension Principle~\cite{Szpilrajn30}, $\prec$ can be extended to a linear ordering $\prec^\prime$ on $\evts$.
	We claim that $\prec^\prime$ itself is the desired well ordering.
	For each $E_i \in \Pi$, $\prec^\prime \bigcap \ (\evts_i \times \evts_i) = \ \prec \bigcap \ (\evts_i \times \evts_i)$
	since the latter is already a well ordering.
	Let $A$ be a nonempty subset of $\evts$.
	$\Pi^\prime = \{A_i \mybigvert (A_i = A \bigcap \evts_i) \land (A_i \text{ is not empty})\}$ is a partition of $A$.
	Assume $\arrowvert \Pi^\prime \arrowvert = k \in [1, m]$.
	In each part of $A$, there exists a least element $e_{j \in [1, k]}$ under $\prec^\prime$.
	The set $\{e_1, e_2, \cdots, e_k\}$ is a nonempty finite set,
	so it has a least element $e^\star$ which is apparently the least of $A$.
\end{proof}

Now we define a concurrent execution by a \textit{history}:
\begin{definition}\label{def_his}
	A history $\his$ is a countable well-ordered structure $\langle \evts, \prec \rangle$ of invocation and response events such that:
	\begin{enumerate}[(1)]
		\item
		There exists a \textit{unique} finite partition $\Pi = \{\evts_1, \evts_2, \cdots, \evts_m\}$ of $\evts$
		where events in each $\evts_i$ share the same process id $i$.
		\item
		It is \textit{well-formed} such that for each $\evts_i \in \Pi$:
		\begin{enumerate}[(a)]
			\item
			The least element in each $\evts_i$ is an invocation event.
			\item
			Let $e_1 \in \evts_i$, and $e_2$ (if existing) be the least of $\evts_i \, \bigcap \, \{e \, \myvert \, e_1 \prec e\}$ under $\prec$:
			\begin{enumerate}[(i)]
				\item
				If $e_1$ is an invocation event, then $e_2$ is the matching response event.
				\item
				If $e_1$ is a response event, then $e_2$ is an invocation event.
			\end{enumerate}
		\end{enumerate}
	\end{enumerate}
\end{definition}
Such history definition over a simple sequence by previous work has advantages of simplicity and extensibility 
in later discussion on history properties and operations.

For brevity, $inv$ and $resp$ stand for an invocation event and a response event respectively.
$e$ or $e^\prime$ is either type of events.
If $inv$ and $resp$ are matching invocation and response events,
then $m = \langle \{inv, resp\}, \{\langle inv, resp\rangle\} \rangle$ represents a method call,
and $\myinv(m) = inv$, $\myresp(m) = resp$.
$\mym(e)$ is $e$'s belonging method call.
{\his}, $\his^\prime$ are histories.
$\evts_{\his}$ is the event set of {\his}, while $\prec_{\his}$ is the corresponding well ordering on $\evts_{\his}$.
We say $e \in \his$ if $e \in \evts_{\his}$, and $m \in \his$ if $\myinv(m),\, \myresp(m) \in \evts_{\his}$.
$o$ stands for an object, and $p$ for a process.

Two histories $\his$ and $\his^\prime$ are \textit{equal}
if $\evts_{\his} = \evts_{\his^\prime}$ and $\prec_{\his} = \prec_{\his^\prime}$.
$\his$ is a \textit{subhistory} of $\his^\prime$
if $(\evts_{\his} \subseteq \evts_{\his^\prime}) \land (\prec_{\his} \subseteq \prec_{\his^\prime})$.
Conversely, $\his^\prime$ is called an \textit{extension} of $\his$.
We denote such subhistory-extension relationship as $\his \subseteq \his^\prime$.
If $\evts \subseteq \evts_{\his^\prime}$,
then $\his = \langle \evts, \prec_{\his^\prime} \bigcap \ (\evts \times \evts) \rangle$ is a subhistory of $\his^\prime$,
because $\prec_{\his^\prime} \bigcap \ (\evts \times \evts)$ is a well ordering~\cite{Enderton77}
and $\his$ apparently meets the two requirements of history definition.
For $\his \subseteq \his^\prime$, we define their \textit{difference}
$\his^\prime - \his = \langle (\evts_{\his^\prime} \setminus \evts_{\his}), \prec_{\his^\prime} \bigcap \ ((\evts_{\his^\prime} \setminus \evts_{\his}) \times (\evts_{\his^\prime} \setminus \evts_{\his})) \rangle$.
$\his^\prime - \his$ is a history too.
We define the \textit{concatenation} function on histories by
$\his_1 \ast \his_2 = \langle \evts_{\his_1} \bigcup \evts_{\his_2}, \prec_{\his_1} \bigcup \prec_{\his_2} \bigcup \, (\prec_{\his_1} \times \prec_{\his_2})\rangle$.
The resulted structure is still well-ordered~\cite{Enderton77} to be a history.
$\his^{i \in \mathbb{N}}$ is an abbreviation to an empty history $\his_{\emptyset}$ if $i = 0$,
or $\his \ast \his \cdots \ast \his$ for $i$ times.

An invocation event is \textit{pending} if its matching response event does not exist.
A method call is \textit{pending} if its invocation event is pending;
it is \textit{complete} if both invocation and response events are in the history.
Function \cmplt(\his) is the maximal subhistory of {\his} without any pending invocation events.
A history {\his} is \textit{complete} if \cmplt({\his}) = {\his}.

From the well ordering on events, we induce a partial well ordering on method calls such that such that $m \prec_{call} m^\prime$ if $\myresp(m) \prec \myinv(m^\prime)$.
Here, $m^\prime$ may be pending.
Two method calls are \textit{concurrent} if neither's response event (if there exists one) precedes the invocation event of the other.
A history {\his} is \textit{sequential} if:
\begin{enumerate}[(1)]
	\item
	The least event of {\his} is an invocation event.
	\item
	Let $e_2$ be the least of $\{e \, \myvert \, e_1 \prec e\}$ under $\prec$:
	if $e_1$ is an invocation event, then $e_2$ is the matching response;
	if $e_1$ is a response event, then $e_2$ is an invocation.
\end{enumerate}
It is obvious that there will be at most one pending method call in a sequential history, and the partial well ordering of method calls is indeed a well ordering.
A history is \textit{concurrent} if it is not sequential.

A \textit{process subhistory} $\his \myvert p$ is the maximal subhistory of {\his} in which all events are generated by process $p$.
Then the well-formed requirement of history definition can be restated as: every process subhistory is sequential.
An \textit{object subhistory} $\his \myvert o$ is defined similarly for an object $o$.
{\his} is a \textit{single-object} history written as $\his_o$ if there exists $o$ such that $\his \myvert o = \his$.
Two histories $\his$ and $\his^\prime$ are \textit{equivalent}
if for each process $p$, $\his \myvert p = \his^\prime \myvert p$.
$\his$ and $\his^\prime$ are \textit{effect-equivalent} in a certain system state $s_1$
if there exists another state $s_2$ such that
$s_1 \stackrel{\his}{\longrightarrow} s_2 \iff s_1 \stackrel{\his^\prime}{\longrightarrow} s_2$.

We adopt the techniques in~\cite{Herlihy90} to define history correctness.
$\his$ is a \textit{prefix} of $\his^\prime$ if $\his \subseteq \his^\prime$
and for all events $e \in \his$ and $e^\prime \in \his^\prime - \his$ it holds that $e \prec_{\his^\prime} e^\prime$.
A set of histories is \textit{prefix-closed} if whenever {\his} is in the set, every prefix of {\his} are also the members. 
A \textit{sequential specification} for an object is a prefix-closed set of single-object sequential histories for that object.
A sequential history {\his} is \textit{legal} if for each object $o$, $\his \myvert o$ belongs to $o$'s sequential specification.

\section{Definition of Linearizability}\label{sec:def_lin}
\begin{definition}\label{def_lin}
	A history $\his = \langle \evts, \prec \rangle$ is \textit{\linb} if:
	\begin{description}
		\item[L1:]
		{\his} can be extended to a new history $\his^\prime = \langle \evts^\prime, \prec^\prime \rangle$ such that $(\his^\prime - \his)$ contains only response events.
		\item[L2:]
		\cmplt($\his^\prime$) is equivalent to some legal sequential history {\shis}.
		\item[L3:]
		For two different method calls $m$ and $m^\prime$, if $m \prec_{\cmplt(\his^\prime)} m^\prime$, then $m \prec_{\shis} m^\prime$.
	\end{description}
\end{definition}

Pending method calls in the original history may or may not have taken effects.
Those having taken effects are captured by extending {\his} with  future matching response events.
Later restriction to \cmplt($\his^\prime$) eliminates remaining pending method calls without real impact on the system.
Original definition of {\lin}~\cite{Herlihy90} takes event ordering inclusion $\prec_{\his} \subseteq \prec_{\shis}$
as condition \textbf{L3}.
However, it is not rigorous and misses an intricate case in which a pending method call in {\his} without real impact
is excluded from \cmplt($\his^\prime$) and thus from {\shis}.
That is why we take $\prec_{\cmplt(\his^\prime)}$ instead in the third condition.

The above legal sequential history {\shis} is called a \textit{\linz} of $\his$.
$\his$ may have more than one \linz.
E.g. if two concurrent complete method calls $m, m^\prime$ operate on different objects,
either one precedes the other could be legal in the final {\linz}.

\section{Proof on the Compositionality of Linearizability}\label{sec:comp}

\subsection{History Properties}\label{sec:hisproperty}
In this subsection, we deduce some properties of history for later usage in proving the compositionality of
linearizability.

\begin{prop}\label{lem_his_eq}
	Two histories $\his$ and $\his^\prime$ are equal if
	\begin{enumerate}[(1)]
		\item $e \in \his \iff e \in \his^\prime$
		\item\label{lem_his_eq_c2} $e \prec_{\his} e^\prime \iff e \prec_{\his^\prime} e^\prime$
	\end{enumerate}
\end{prop}
\begin{proof}
	This is a restatement of history equality.
	Condition 1 implies $\evts_{\his} = \evts_{\his^\prime}$ while condition~\ref{lem_his_eq_c2} implies $\prec_{\his} = \prec_{\his^\prime}$.
\end{proof}

\begin{prop}\label{lem_schis_eq}
	Two sequential and complete histories $\his$ and $\his^\prime$ are equal if
	\begin{enumerate}[(1)]
		\item\label{lem_schis_eq_c1}
		$m \in \his \iff m \in \his^\prime$
		\item\label{lem_schis_eq_c2}
		$m \prec_{\his} m^\prime \iff m \prec_{\his^\prime} m^\prime$
	\end{enumerate}
\end{prop}
\begin{proof}
	\begin{enumerate}[(a)]
		\item
		\begin{align*}
		e \in \his &\iff (\mym(e) \text{ is complete}) \land (\mym(e) \in \his)\\
		&\iff \mym(e) \in \his^\prime \\
		&\iff e \in \his^\prime
		\end{align*}
		
		\item
		Assume $e \prec_{\his} e^\prime$, then there are two cases:
		\begin{enumerate}[(i)]
			\item
			$\mym(e) = \mym(e^\prime)$. That is, $e$ and $e^\prime$ are matching events. Then $\mym(e)$ is also in $\his^\prime$ and $e \prec_{\his^\prime} e^\prime$.
			\item
			$\mym(e) \neq \mym(e^\prime)$. That is, $e$ and $e^\prime$ belong to two different complete method calls.
			There are four subcases on the types of $e$ and $e^\prime$. In each subcase, we have $m \prec_{\his} m^\prime$
			since a response event is the immediate successor of its matching invocation event in a sequential history.
			By condition~\ref{lem_schis_eq_c2}, we get $m \prec_{\his^\prime} m^\prime$ and then $e \prec_{\his^\prime} e^\prime$.
		\end{enumerate}	
		In both cases we have $e \prec_{\his} e^\prime \implies e \prec_{\his^\prime} e^\prime$.
		Similarly, we have $e \prec_{\his^\prime} e^\prime \implies e \prec_{\his} e^\prime$.
	\end{enumerate}
	By Proposition \ref{lem_his_eq}, we conclude that $\his$ and $\his^\prime$ are equal.
\end{proof}

\begin{prop}\label{lem_chis_1}
	\cmplt({\his}) is complete.
\end{prop}
\begin{proof}
	\cmplt({\his}) does not have pending invocation events any more, therefore its maximal subhistory without pending invocation events is itself.
\end{proof}

\begin{prop}\label{lem_chis_2}
	If {\his} is complete, then both $\his \myvert o$ and $\his \myvert p$ are complete.
\end{prop}
\begin{proof}
	$\his \myvert o$ and $\his \myvert p$ do not have any pending invocation events.
\end{proof}

\begin{prop}\label{lem_his_eqiv}
	If $\his$ and $\his^\prime$ are equivalent, then $m \in \his \iff m \in \his^\prime$.
\end{prop}
\begin{proof}
	\begin{align*}
	m \in \his &\iff \exists p \ (m \in \his \myvert p)\\
	&\iff \exists p \ (m \in \his^\prime \myvert p)\\
	&\iff m \in \his^\prime
	\end{align*}
\end{proof}

\begin{prop}\label{lem_subhis_el}
	If $\his \subseteq \his^\prime$ and $e$, $e^\prime \in \his$,
	then $e \prec_{\his} e^\prime \iff e \prec_{\his^\prime} e^\prime$.
\end{prop}
\begin{proof}
	By the definition of subhistory.
\end{proof}

\begin{prop}\label{lem_subhis_ml}
	If $\his \subseteq \his^\prime$ and $m$, $m^\prime \in \his$,
	then $m \prec_{\his} m^\prime \iff m \prec_{\his^\prime} m^\prime$.
\end{prop}
\begin{proof}
	By Proposition \ref{lem_subhis_el}.
\end{proof}

\begin{prop}\label{lem_diffhis_oeq}
	If $\his \subseteq \his^\prime$, then $(\his^\prime - \his) \myvert o = \his^\prime \myvert o - \his \myvert o$.
\end{prop}
\begin{proof}
	\begin{enumerate}[(a)]
		\item
		It is obvious that $\his \myvert o \subseteq \his^\prime \myvert o$.
		
		\item
		\begin{align*}
		e \langle i, j, x, y, \_ \rangle \in (\his^\prime - \his) \myvert o &\iff (o_x = o) \land (e \langle i, j, x, y, \_ \rangle \in \his^\prime)\\
		&\land (e \langle i, j, x, y, \_ \rangle \notin \his) \\
		&\iff (e \langle i, j, x, y, \_ \rangle \in \his^\prime \myvert o) \land (e \langle i, j, x, y, \_ \rangle \notin \his \myvert o) \\
		&\iff e \langle i, j, x, y, \_ \rangle \in \his^\prime \myvert o - \his \myvert o
		\end{align*}
		
		\item
		\begin{align*}
		e \prec_{(\his^\prime - \his) \myvert o} e^\prime &\iff (e, e^\prime \in {(\his^\prime - \his) \myvert o}) \land (e \prec_{\his^\prime - \his} e^\prime)\\
		&\iff (e, e^\prime \in \his^\prime \myvert o - \his \myvert o ) \land (e \prec_{\his^\prime - \his} e^\prime)\\
		&\iff (e, e^\prime \in \his^\prime \myvert o) \land (e, e^\prime \notin \his \myvert o) \land (e \prec_{\his^\prime} e^\prime) \\
		&\iff (e, e^\prime \in \his^\prime \myvert o) \land (e, e^\prime \notin \his \myvert o) \land (e \prec_{\his^\prime \myvert o} e^\prime) \\
		&\iff e \prec_{\his^\prime \myvert o - \his \myvert o} e^\prime
		\end{align*}	
	\end{enumerate}
\end{proof}

\begin{prop}\label{lem_inv}
	Assume $\his \subseteq \his^\prime$, $\his^\prime - \his$ contains only response events and $e$, $e^\prime$ are invocation events, then:
	\begin{enumerate}[(1)]
		\item
		$e \in \his \iff e \in$ \cmplt({\his}) $\iff e \in \his^\prime \iff e \in$ \cmplt($\his^\prime$)
		\item
		$e \prec_{\his} e^\prime \iff e \prec_{\cmplt(\his)} e^\prime \iff e \prec_{\his^\prime} e^\prime \iff e \prec_{\cmplt(\his^\prime)} e^\prime$
	\end{enumerate}
\end{prop}
\begin{proof}
	\cmplt({\his}), $\his^\prime$ and \cmplt($\his^\prime$) do not affect the orders between invocation events of the original {\his}.
\end{proof}

\begin{prop}\label{lem_subhis_op}
	$(\his \myvert o) \myvert p = (\his \myvert p) \myvert o$
\end{prop}
\begin{proof}
	\begin{enumerate}[(a)]
		\item
		\begin{align*}
		e \langle i, j, x, y, \_ \rangle \in (\his \myvert o) \myvert p &\iff (p_i = p) \land (e \langle i, j, x, y, \_ \rangle \in \his \myvert o)\\
		&\iff (p_i = p) \land (o_x = o) \land (e \langle i, j, x, y, \_ \rangle \in \his)\\
		&\iff (e \langle i, j, x, y, \_ \rangle \in \his \myvert p) \land (o_x = o)\\
		&\iff e \langle i, j, x, y, \_ \rangle \in (\his \myvert p) \myvert o
		\end{align*}
		\item
		\begin{align*}
		e \prec_{(\his \myvert o) \myvert p} e^\prime &\iff e \prec_{\his \myvert o} e^\prime \iff e \prec_{\his} e^\prime\\
		&\iff e \prec_{\his \myvert p} e^\prime \iff e \prec_{(\his \myvert p) \myvert o} e^\prime
		\end{align*}
	\end{enumerate}
	Therefore $(\his \myvert o) \myvert p = (\his \myvert p) \myvert o$.
\end{proof}

\begin{prop}\label{lem_subhis_complo}
	$\cmplt(\his \myvert o) = \cmplt(\his) \myvert o$.
\end{prop}
\begin{proof}
	\begin{enumerate}[(a)]
		\item
		\begin{align*}
		e \langle i, j, x, y, \_ \rangle \in \cmplt(\his \myvert o) &\iff (\mym(e) \text{ is complete}) \land (o_x=o)\\
		&\iff (e \langle i, j, x, y, \_ \rangle \in \cmplt(\his)) \land (o_x=o)\\
		&\iff e \langle i, j, x, y, \_ \rangle \in \cmplt(\his) \myvert o
		\end{align*}
		
		\item
		Because \cmplt({\his}) and \cmplt($\his \myvert o$) are also subhistories of {\his}, we have
		\begin{align*}
		e \prec_{\cmplt(\his \myvert o)} e^\prime &\iff e \prec_{\his \myvert o} e^\prime \iff e \prec_{\his} e^\prime\\
		&\iff e \prec_{\cmplt(\his)} e^\prime \iff e \prec_{\cmplt(\his) \myvert o} e^\prime
		\end{align*}
	\end{enumerate}
	Therefore \cmplt($\his \myvert o$) = \cmplt({\his})$\myvert o$.
\end{proof}

\subsection{The Compositionality Theorem}\label{comp_thm}

\begin{theorem}\label{thm_lin}
	{\his} is {\linb} if and only if, for each object $o$, $\his \myvert o$ is {\linb}.
\end{theorem}

We prove ``only if'' and ``if'' parts separately.

\begin{lemma}\label{lemma_lin_onlyif}
	{\his} is {\linb} only if for each object $o$, $\his \myvert o$ is {\linb}.
\end{lemma}
\begin{proof}
	Suppose {\shis} is one {\linz} of history {\his}, and $\his^\prime$ is the corresponding extension.
	We claim that $\shis \myvert o$ is a {\linz} of $\his \myvert o$ whose extension is just $\his^\prime \myvert o$.
	\begin{enumerate}[(a)]
		\item
		$\his^\prime \myvert o$ is a history since $\his^\prime$ is a history, and it is obviously that $\his \myvert o \subseteq \his^\prime \myvert o$.
		Since $\his^\prime \myvert o - \his \myvert o = (\his^\prime - \his) \myvert o$ by Proposition \ref{lem_diffhis_oeq} and $\his^\prime - \his$ contains only response events, $\his^\prime \myvert o - \his \myvert o$ contains only response events too.
		\item
		{\shis} is a legal, sequential and complete history, therefore $\shis \myvert o$ is also a legal, sequential and complete history.
		\item
		Because \cmplt($\his^\prime$) is equivalent to {\shis},
		\cmplt($\his^\prime$)$\myvert p = \shis \myvert p$.
		\item
		\begin{align*}
		\cmplt(\his^\prime \myvert o) \myvert p &= (\cmplt(\his^\prime) \myvert o) \myvert p \quad \text{by Proposition \ref{lem_subhis_complo}}\\
		&= (\cmplt(\his^\prime) \myvert p) \myvert o \quad \text{by Proposition \ref{lem_subhis_op}}\\
		&= (\shis \myvert p) \myvert o = (\shis \myvert o) \myvert p
		\end{align*}
		Thus $\cmplt(\his^\prime \myvert o)$ is equivalent to $\shis \myvert o$.
		\item
		\begin{align*}
		m \prec_{\cmplt(\his^\prime \myvert o)} m^\prime &\implies m \prec_{\cmplt(\his^\prime) \myvert o} m^\prime\\
		&\implies m \prec_{\cmplt(\his^\prime)} m^\prime\\
		&\implies m \prec_{\shis} m^\prime\\
		&\implies m \prec_{\shis \myvert o} m^\prime \quad \text{since } m, m^\prime \in \shis \myvert o
		\end{align*}
	\end{enumerate}
\end{proof}

\begin{lemma}\label{lemma_lin_if}
	{\his} is {\linb} if for each object $o$, $\his \myvert o$ is {\linb}.
\end{lemma}

The original proof~\cite{Herlihy86, Herlihy87, Herlihy90, Herlihy08} of this lemma is a proof of existence,
by demonstrating a global partial ordering among object {\linz}s exists.
Our idea, inspired by Multiway Merge Algorithm~\cite{Cormen}, is to construct the final {\linz} directly without building such global partial ordering.
The following Algorithm~\ref{alg:lingen} shows the details for finite histories.
The algorithmic step \ref{alg:lingen_sel} always succeeds because invocation events in each $\shis_o$ are also in {\his}
by Proposition \ref{lem_inv},
and events in {\his} are already well-ordered.

\begin{algorithm}
\renewcommand{\algorithmicrequire}{\textbf{Input:}}
\renewcommand{\algorithmicensure}{\textbf{Output:}}

\caption{Construct a {\linz} from object {\linz}s}
\label{alg:lingen}

\begin{algorithmic}[1]

\REQUIRE ~~\\
$n$ objects $\{o_1, o_2, \cdots, o_n\}$;

A history ${\his}$ and its object subhistories $\{\his_{o_1}, \his_{o_2}, \cdots, \his_{o_n}\}$;

{\Linz}s of object subhistories $\{\shis_{o_1}, \shis_{o_2}, \cdots, \shis_{o_n}\}$;

\smallskip
\ENSURE A {\linz} {\shis} of {\his};

\bigskip
\STATE INITIALIZE $\shis = \langle \emptyset, \emptyset \rangle$

\WHILE{some $\shis_{o} \neq \langle \emptyset, \emptyset \rangle$}
	 \STATE{Select a method call $m \in \shis_{o_i}$, whose invocation event is least in {\his}}\label{alg:lingen_sel}
	 \STATE{$\shis_{o_i} \gets \shis_{o_i} - m$}
	 \STATE{$\shis \gets \shis \ast m$}
\ENDWHILE
 
\end{algorithmic}
\end{algorithm}

We briefly discuss the asymptotic time and space complexity.
Suppose there are P processes and N events in {\his} in which P $\ll$ N.
Although there may be $r \in$ [0, P] pending invocation events in {\his},
we assume the total event number of all {\linz}s of object subhistories equals to N for simplicity since P $\ll$ N.
Events are associated with unique integer time stamps (may be P-ary vector clocks in practice) to determine their order.
Thus, our algorithm finishes in $\mathcal{O}(\text{N}\times\log{}\text{P})$ running time with $\mathcal{O}(\text{N})$ space.
A typical algorithm in the original compositionality proof has two steps:
the global partial ordering construction using a DAG runs in $\mathcal{O}(\text{N}^{2})$ time with $\mathcal{O}(\text{N} + \mid\text{Edges}\mid)$ space, in which $\mid\text{Edges}\mid \in [\frac{\text{N}}{2}-\text{P}, \frac{\text{N}\times(\text{N} - 2)}{8}]$;
the topological sorting runs in $\mathcal{O}(\text{N} + \mid\text{Edges}\mid)$ time with $\mathcal{O}(\text{N} + \mid\text{Edges}\mid)$ space.

\begin{proof}
	Let $(\his \myvert o)^\prime$
	be an extension of $\his \myvert o$ such that $(\his \myvert o)^\prime - \his \myvert o$ contains only response events, and $\shis_o$ be the corresponding {\linz}.
	We construct a sequential history $\shis = \langle \evts_{\shis}, \prec_{\shis} \rangle$ as:
	\begin{align*}
	\evts_{\shis} &= \bigcup\limits_o \, \evts_{\shis_o}\\
	\prec_{\shis} &= (\bigcup\limits_o \, \prec_{\shis_o}) \, \bigcup \, \{\langle e, e^\prime \rangle \mybigvert \exists m, m^\prime, o_i, o_j\\
	&\qquad ((e \in m \in \shis_{o_i}) \land (e^\prime \in m^\prime \in \shis_{o_j}) \land (i \neq j) \land (\myinv(m) \prec_{\his} \myinv(m^\prime)))\}
	\end{align*}
	
	Beware the fact that events in $(\his \myvert o_i)^\prime - \his \myvert o_i$ have no causality with those in both $(\his \myvert o_j)^\prime - \his \myvert o_j$ and $\his \myvert o_j$ if $i \neq j$.
	Therefore, we could give them an arbitrary ordering.
	Now we construct the extension history $\his^\prime$:
	\begin{align*}
	\evts_{\his^\prime} &= \bigcup\limits_o \, \evts_{(\his \myvert o)^\prime}\\
	\prec_{\his^\prime} &= (\bigcup\limits_o \, \prec_{(\his \myvert o)^\prime}) \, \bigcup \, \{\langle e, e^\prime \rangle \mybigvert (e \in \his \myvert o_i) \land (e^\prime \in (\his \myvert o_j)^\prime - \his \myvert o_j) \land (i \neq j)\}\\
	&\qquad \bigcup \, \{\langle e, e^\prime \rangle \mybigvert (e \in (\his \myvert o_i)^\prime - \his \myvert o_i) \land (e^\prime \in (\his \myvert o_j)^\prime - \his \myvert o_j) \land (i < j)\}
	\end{align*}
	
	We have the following facts on the constructed {\shis} and $\his^\prime$:
	\begin{enumerate}[(a)]
		\item
		{\shis} is a history. It is obvious that $\prec_{\shis}$ is a linear ordering.
		Suppose $A$ is a non-empty subset of $\evts_{\shis}$.
		Then $\Pi = \{A_i \mybigvert (A_i = A \bigcap \evts_{\shis_{o_i}}) \land (A_i \text{ is not empty})\}$ is a partition of $A$
		because all $\evts_{\shis_{o_i}}$ are pair-wisely disjoint.
		The set $\{e \mybigvert e = \myinv(\mym(e_i)) \text{ where } e_i \text{ is the least of } A_i\}$ does exist since each $A_i$ is well-ordered under $\prec_{\shis_{o_i}}$.
		Let $e^\star$ be the least event of such finite set.	
		Then, either $e^\star$ or its matching response event is the least element of $A$ depending on whether $e^\star \in A$ or not.
		Thus $\prec_{\shis}$ is a well ordering.
		
		Similarly, $\his^\prime$ is a history too.
		\item
		$\shis \myvert o = \shis_o$ and $\his^\prime \myvert o = (\his \myvert o)^\prime$
		\item
		$m \in \shis$ if and only if there exists an object $o$ such that $m \in \shis_o$
		\item
		$m \in \cmplt(\his^\prime)$ if and only if $m \in \shis$ because:
		\begin{align*}
		m \in \cmplt(\his^\prime) &\iff \exists o \ (m \in \cmplt(\his^\prime) \myvert o)\\
		&\iff \exists o \ (m \in \cmplt(\his^\prime \myvert o)) \quad \text{by Proposition \ref{lem_subhis_complo}}\\
		&\iff \exists o \ (m \in \cmplt((\his \myvert o)^\prime))\\
		&\iff \exists o \ (m \in \shis_o) \quad \text{since $\cmplt((\his \myvert o)^\prime)$ and $\shis_o$ are equivalent}\\
		&\iff m \in \shis \quad \text{by the construction of \shis}
		\end{align*}
	\end{enumerate}
	
	We claim that {\shis} is the {\linz} of {\his}, and $\his^\prime$ is the corresponding extension.
	
	First, $\evts_{\his^\prime - \his} = (\bigcup\limits_o \, \evts_{(\his \myvert o)^\prime}) \setminus \evts_{\his} = \bigcup\limits_o \, (\evts_{(\his \myvert o)^\prime} \setminus \evts_{\his \myvert o}) = \bigcup\limits_o \, \evts_{(\his \myvert o)^\prime - \his \myvert o}$. This is a set containing only response events.
	
	Second, {\shis} is sequential because:
	\begin{enumerate}[(a)]
		\item
		The least event $e^\star$ of {\shis} is also the least one of some $\shis_o$.
		Since $\shis_o$ is a sequential history, $e^\star$ is an invocation event.
		\item
		Let $m = \langle inv, resp \rangle$ be a complete method call.
		Suppose there exists a third event $e \in m^\prime$ that $inv \prec_{\shis} e \prec_{\shis} resp$.
		Then $m$ and $m^\prime$ are from different $\shis_o$ since each $\shis_o$ is a sequential history.
		However, $inv \prec_{\shis} e$ implies that $inv \prec_{\his} \myinv(m^\prime)$, then $resp \prec_{\shis} e$ by the construction of {\shis}.
		Such contradiction explains that the immediate successor of an invocation event must be its matching response event.
		\item
		Let $e$ be a response event and $e^\prime$ be its immediate successor in {\shis}.
		Let $m$ and $m^\prime$ be the two events' belonging method calls.
		If $m$ and $m^\prime$ are from the same $\shis_o$, then $e^\prime$ must be an invocation event since $\shis_o$ is sequential.
		If $m$ and $m^\prime$ are from different $\shis_o$ and $e^\prime$ is a response event,
		$e \prec_{\shis} e^\prime$ implies that $\myinv(m) \prec_{\his} \myinv(m^\prime)$,
		then $e \prec_{\shis} \myinv(m^\prime) \prec_{\shis} e^\prime$.
		Such contradiction means that the immediate successor of a response event must be an invocation event.
	\end{enumerate}
	
	Third, {\shis} is a legal history because {\shis} is sequential and $\shis \myvert o = \shis_o$ which is the {\linz} of a single-object history.
	
	Fourthly, for two different method calls $m$ and $m^\prime$, if $m \prec_{\cmplt(\his^\prime)} m^\prime$, then $m \prec_{\shis} m^\prime$:
	\begin{enumerate}[(a)]
		\item
		Both $m$ and $m^\prime$ are on the same object:
		\begin{align*}
		m \prec_{\cmplt(\his^\prime)} m^\prime &\implies \exists o \ ((m, m^\prime \in \cmplt(\his^\prime) \myvert o) \land (m \prec_{\cmplt(\his^\prime)} m^\prime))\\
		&\implies m \prec_{\cmplt(\his^\prime) \myvert o} m^\prime\\
		&\implies m \prec_{\cmplt(\his^\prime \myvert o)} m^\prime\\
		&\implies m \prec_{\cmplt((\his \myvert o)^\prime)} m^\prime\\
		&\implies m \prec_{\shis_o} m^\prime \quad \text{since $\shis_o$ is a {\linz}}\\
		&\implies m \prec_{\shis} m^\prime
		\end{align*}
		\item
		$m$ and $m^\prime$ are on different objects:
		\begin{align*}
		m \prec_{\cmplt(\his^\prime)} m^\prime &\implies \exists o_i, o_j \ ((i \neq j) \land (m \in \cmplt(\his^\prime) \myvert o_i)\\
		&\qquad \land (m^\prime \in \cmplt(\his^\prime) \myvert o_j) \land (m \prec_{\cmplt(\his^\prime)} m^\prime))\\
		&\implies \myinv(m) \prec_{\cmplt(\his^\prime)} \myresp(m) \prec_{\cmplt(\his^\prime)} \myinv(m^\prime)\\
		&\implies \myinv(m) \prec_{\his} \myinv(m^\prime) \quad \text{by Proposition \ref{lem_inv}}\\
		&\implies \myinv(m) \prec_{\shis} \myresp(m) \prec_{\shis} \myinv(m^\prime) \prec_{\shis} \myresp(m^\prime)\\
		&\implies m \prec_{\shis} m^\prime
		\end{align*}
	\end{enumerate}
	
	Finally, {\shis} is equivalent to $\cmplt(\his^\prime)$ because:
	\begin{enumerate}[(a)]
		\item
		$\cmplt(\his^\prime) \myvert p$ and $\shis \myvert p$ are sequential and complete histories.
		\item
		\begin{align*}
		m \langle i, j, x, y, \_ \rangle \in \cmplt(\his^\prime) \myvert p &\iff (p_i = p) \land (m \langle i, j, x, y, \_ \rangle \in \cmplt(\his^\prime))\\
		&\iff (p_i = p) \land (m \langle i, j, x, y, \_ \rangle \in \shis)\\
		&\iff m \langle i, j, x, y, \_ \rangle \in \shis \myvert p
		\end{align*}
		\item
		\begin{align*}
		&m \langle i, j, x, y, \_ \rangle \prec_{\cmplt(\his^\prime) \myvert p} m^\prime \langle r, s, u, v, \_ \rangle\\
		&\iff (m \langle i, j, x, y, \_ \rangle \prec_{\cmplt(\his^\prime)} m^\prime \langle r, s, u, v, \_ \rangle) \land (p_i = p_r = p)\\
		&\iff (m \langle i, j, x, y, \_ \rangle \prec_{\shis} m^\prime \langle r, s, u, v, \_ \rangle) \land (p_i = p_r = p)\\
		&\iff m \langle i, j, x, y, \_ \rangle \prec_{\shis \myvert p} m^\prime \langle r, s, u, v, \_ \rangle
		\end{align*}
		We complete the proof and conclude that {\shis} is indeed the {\linz} of {\his}.
	\end{enumerate}
\end{proof}

\section{Conclusion and Future Work}\label{sec:discussion}
Our contribution is twofold.
First, we propose a new hierarchical system model to reason about concurrent properties not limited to {\lin}.
Histories are defined as countable well-ordered structures of events
to deal with both finite and infinite executions, and reveal precisely the internal event ordering.
This definition has advantages of simplicity and extensibility over a simple sequence. 
Second, we present a new constructive proof on the compositionality theorem of {\lin}.
The proof deduces a more theoretically efficient Multiway Merge alike algorithm to generate a final {\linz} without the need to build the global partial ordering.

Since objects are defined in a hierarchical way by induction,
it is possible to unwind a history on high level objects to another one on low level objects,
or to fold a history on low level objects conversely.
Let $o = \langle \{o_1, o_2, \cdots, o_n\}, \{op_1, op_2, \cdots, op_k\} \rangle$ be a composite object.
Informally, if we replace each method call in $\his \myvert o$ with effect-equivalent consecutive method calls on $o_{i \in [1, n]}$,
we say the resulted $\his^c$ is a concretization of $\his$ on $o$.
On contrary, we aggregate all maximal consecutive method calls on $o_{i \in [1, n]}$ in each $\his \myvert p$ into individual effect-equivalent method calls on $o$,
and the final $\his^s$ is called the summarization of $\his$ on $o$.
To be noted, we need to assure they are histories.

Thus, one interesting and natural future work is to study property preservation across system layers
by history concretization and summarization.
It is not surprising that history summarization may break {\lin} because operation semantics could change across layers.
Find out the exact conditions under which {\lin} holds upward
will guide us in bottom-up system construction and verification~\cite{export:121499, vafeiadis2009shape, Colvin06formalverification, Vafeiadis06provingcorrectness, DBLP:conf/esop/ZomerGRS14}.

Another work is to automatically optimize {\linb} implementation.
After having identified the {\linb} region through history concretization,
code motion technique~\cite{Aho:1986:CPT:6448, nnh2010} may be used to move program statements unrelated to {\lin} out of such region to boost synchronization performance (e.g. people use coarse-grained locking).
It is also promising to apply our methodology and the hierarchical viewpoint to other properties
such as sequential consistency~\cite{Lamport:1979:MMC:1311099.1311750},
eventual consistency~\cite{Vogels:2009:EC:1435417.1435432, Serafini:2010:ELS:1835698.1835723},
wait-freedom~\cite{Herlihy:1991:WS:114005.102808} and obstruction-freedom~\cite{Herlihy:2003}.

\bibliographystyle{plainurl}
\bibliography{main}

%------------------------------------------------------------------------------

\end{document}